\documentclass[letterpaper, 10pt,  peerreview]{IEEEtran} 
\normalsize

\DeclareMathAlphabet\mathbfcal{OMS}{cmsy}{b}{n}

\usepackage{amsthm}
\usepackage{algorithm}
\usepackage{algorithmic}
\usepackage{graphicx,amsmath,amssymb,latexsym,epsfig}
\usepackage{setspace}
\usepackage{psfrag}
\usepackage{array,booktabs,arydshln,xcolor}

\usepackage{lscape}
\newtheorem{theorem}{Theorem}
\newtheorem{problem}{Problem}
\newtheorem{corollary}{Corollary}
\newtheorem*{proof of Theorem*}{Proof of Theorem 3}
\newtheorem{proof of Lemma}{Proof of Lemma}
\newtheorem{definition}{Definition}
\newtheorem{lemma}{Lemma}
\newtheorem{remark}{Remark}

\newcommand{\enp} {\hfill \rule{2.2mm}{2.6mm}}

\begin{document}


\title{Optimal Energy Allocation Policies for a High Altitude Flying Wireless Access Point}

\author{\IEEEauthorblockN{Elif Tugce Ceran, Tugce Erkilic, Elif Uysal-Biyikoglu, Tolga Girici and Kemal Leblebicioglu}
\IEEEauthorblockA{E-mail: e.ceran14@imperial.ac.uk
, terkilic@aselsan.com.tr, uelif@metu.edu.tr, tgirici@etu.edu.tr, kleb@metu.edu.tr}}




\bibliographystyle{wileyj}


\maketitle

\def\eg{\emph{e.g.}}
\def\ie{\emph{i.e.}}

\begin{abstract}
 Inspired by recent industrial efforts toward high altitude flying wireless access points powered by renewable energy, an online resource allocation problem for a mobile access point (AP) travelling at high altitude is formulated. The AP allocates its resources (available energy) to maximize the total utility (reward) provided to a sequentially observed set of users demanding service. The problem is formulated as a 0/1 dynamic knapsack problem with incremental capacity over a finite time horizon, the solution of which is quite open in the literature.  We address the problem through deterministic and stochastic formulations. For the deterministic problem, several online approximations are proposed based on an instantaneous threshold that can adapt to short-time-scale dynamics. For the stochastic model, after showing the optimality of a threshold based solution on a dynamic programming (DP) formulation, an approximate threshold based policy is obtained. The performances of proposed policies are compared with that of the optimal solution obtained through DP. \footnote[1]{This is an extended version of the paper ``Optimizing the service policy of a wireless access point on the move with renewable energy" appearing in Proc. of 52nd Annual Allerton Conference on Communication, Control, and Computing (Allerton), Monticello Illinois, pp.967-974, Sept. 30 2014-Oct. 3 2014.} 
\end{abstract} 

\section{Introduction}

The possibility of providing ubiquitous Internet access by filling coverage gaps in rural and remote areas devoid of ground based Infrastructure through deployment of mobile access points (AP) in the Earth’s atmosphere (e.g. floating at 20 km altitude,  the lower stratosphere, see, for example \cite{balloon2014, Loonexample,Droneexample}) has spurred recent industry effort. Such networks with autonomous APs require long term unattended operation without a steady energy source or a change of batteries, and are typically powered by renewable energy sources, particularly by solar energy harvesting \cite{Loonexample,Droneexample}.  

We define an Access Point on the Move (APOM) as a flying platform that provides Internet service to users it encounters along its path, who demand service with possibly different quality requirements. This paper considers the problem of optimal allocation of harvested energy in time by an APOM to a dynamic set of users.

\begin{figure}
\includegraphics[scale=0.6]{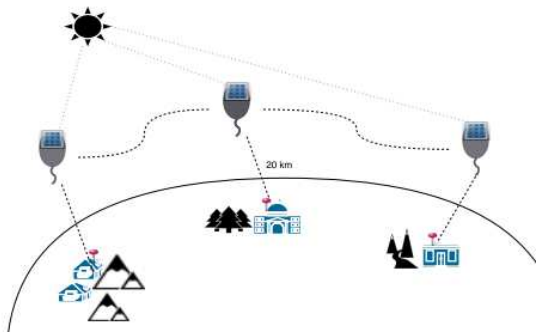}
\centering
\caption{Illustration of an Access Point on the Move which harvests energy and encounters user demands as it flies over a geographical area.}
\end{figure}

We concern ourselves with the following problem: Suppose that, as the APOM moves along its trajectory, it encounters users sequentially. Given its instantaneous energy budget which is replenished by arbitrary amounts at arbitrary times, observed causally, how should it decide to serve or reject a given customer? We impose a utility function related to a total value collected over users that are served within a given time frame. The deterministic version of the problem can be modeled as an online 0/1 knapsack problem with dynamic capacity, where the capacity, or service capacity, at a given time, corresponds to the total energy stored in the battery at that time. The stochastic version, where the arrival processes of energy and user demand are assumed to be semi-Markov processes, may be modeled as a Semi-Markov Decision Process.

In a practical application scenario, there may be a fleet of APOMs travelling over a geographical area, hence one might imagine that a user that was rejected service by an APOM will be served by another one that will come along. In this paper, we limit attention to solving the single access point problem, while the multiple AP case would be an immediate and interesting extension.

\subsection{Related Work}

According to a International Telecommunications Union (ITU) report \cite{ITU2014}, only 40$\%$ of people around world had access to the Internet as of 2014. Motivated by the potential opportunities presented by this gap, major industry players have been pushing for ubiquitous connectivity by deploying airborne mobile APs \cite{balloon2014, Loonexample,Droneexample}. Mobile service providers in general have certain advantages over fixed ones in terms of being more responsive to varying geographic communication demand \cite{Shi2008}. Various studies exist regarding mobile sinks in conventional networks that do not exploit renewable energy \cite{Gao2011, Shi2008}. Some of  these studies focused on determining optimal paths in order to prolong network lifetime \cite{Guo2012, Alkesh2011}.

In recent years, employing energy harvesting (via ambient energy sources such as solar irradiation \cite{Tekbiyik2013}, vibrations \cite{Seccad_paper}, or electromagnetic sources \cite{ETTenergy} etc.) to power transmitters of network devices such as access points has been drawing great attention from the research community. The consideration of mobile access points with energy harvest capability is relatively new. Several of the recent studies on the topic are concerned with finding optimal routing paths \cite{Ren2012}. Xie et al. \cite{Xie2013} address the problem of colocating the mobile service provider on the wireless charging machine with the objective of minimizing energy consumption. The closely related works of Ren and Liang \cite{Ren2013} considered a distributed time allocation method to maximize data collection in energy harvesting sensor networks while defining a scenario of a constrained path with all sensors having renewable energy sources. 

The optimality of a threshold based approach for a binary decision problem to transmit or defer transmision tasks on a Gilbert-Elliot channel supplied by energy harvesting is proved in \cite{Ephremides2012}. Various other works, such as \cite{Tan2014}, and \cite{Zorzi2012}, also consider threshold based solutions for resource allocation problems in energy harvesting systems for stochastic models. In recent work \cite{allerton,elif_thesis, tugce_thesis} resource allocation for solar powered stationary and mobile service providers have been addressed through various different optimization techniques.

The problem at hand can be set up as a 0/1 dynamic online knapsack problem. While the knapsack problem is a well known combinatorial optimization problem \cite{Martello1990},  for which competitive online solutions are limited \cite{Tiedemann14}. In the sequel, Chakrabarty et. al. propose a constant competitive solution to the problem with static and large capacity under certain assumptions \cite{Zhou2008}. The dynamic capacity case, which applies to the setup in this paper, is largely open.

\subsection{Our Contributions}

One aspect that sets this paper apart from other studies concerned with resource allocation at a mobile access point powered through (solar) energy harvesting, is the formulation of an {\emph{online}} user admission problem. The resource allocation problem is mapped to a  multi constraint \textit{0/1 knapsack problem} (KP). After exhibiting the existence of a threshold based optimal policy, scalable and computationally cost effective heuristics are proposed. The performance of these heuristics are studied numerically (through simulations) and a competitive ratio analysis is conducted. 

There are a number of studies implementing the optimization tools utilized in this paper such as genetic algorithms and rule-based logic, in resource allocation literature. However, to the best of our knowledge this is the first application of these techniques to a threshold based user selection problem. Furthermore, our schemes constitute a fairly competitive solution to the dynamic knapsack problem with incremental capacity, optimal competitive ratios for which are currently not available. 

\subsection{Organization of the paper}

The rest of the paper is organized as follows: The system model and deterministic problem formulation are provided in Section \ref{sec:pstat}.  Next, the stochastic version of the problem is examined and several approaches providing both optimal and suboptimal solutions are proposed in Section \ref{sec:stoc}.  After exhibiting an optimal algorithm and the existence of a threshold for the stochastic model, the problem is investigated for a deterministic model where several novel optimization tools are employed to propose threshold based heuristic solutions to the resource allocation problem in Section \ref{sec:deter}.  Detailed numerical and simulation results are presented in Section \ref{sec:results}. The paper is concluded with an outline of future directions in Section \ref{sec:conc}.


\section{System Model and Problem Statement}
\label{sec:pstat}


Consider an APOM moving on a predefined path and empowered by renewable energy resources. As it travels on its route, each user request is regarded as an \textit{event} representing the start of a time \textit{slot}. Decisions are to be made on an \textit{event-based schedule}, i.e. once per slot. We consider a finite horizon problem with $N$ slots and $N$ (distinct) corresponding users. Considering the operating scenario on rural areas, it is reasonable that stationary end users (local base station transceivers) are distributed sparsely over a large geographical area. Due to the topology, it is presumable that one such user will be observed at the time, which motivates a sequential user arrival model. 

Note that each energy replenishment will constitute an increase in the service capacity of the APOM. The problem will be broken down into $J$ harvest periods as depicted in Figure \ref{fig:model}, each incoming (instantaneous) energy replenishment starts a new subinterval. Each subinterval contains an integer number of slots (not necessarily corresponding to equal time intervals). Each slot corresponds to a single user demand being presented to the APOM. A user is characterized by a \textit{value} and \textit{weight} pair:  $(v_n, w_n)$ for the $n^{th}$ user. The \textit{value} of a user corresponds to the utility gained by serving that user whereas the \textit{weight} corresponds to a cost (e.g., the energy consumption, i.e. the reduction from the total instantaneous service capacity of the APOM) related to serving that user. The difference in costs of users may in practice be related to the users having different predefined QoS (Quality of Service) requirements, their channels having different qualities, etc. Differences in utility could model different user or service priorities/types, as well as payments.   

The APOM makes a binary decision at each slot, whether to serve the encountered user or not. Once a decision on a user has been made, there will be no re-evaluation of the same user, therefore the decisions are irreversible. In the case of interest where energy replenishment rate cannot meet the power needed to serve all incoming users,  the APOM has to  pick a proper subset of users to maximize total utility under energy constraints in an online fashion. 
\begin{figure*}
\centering
\includegraphics[scale=0.7]{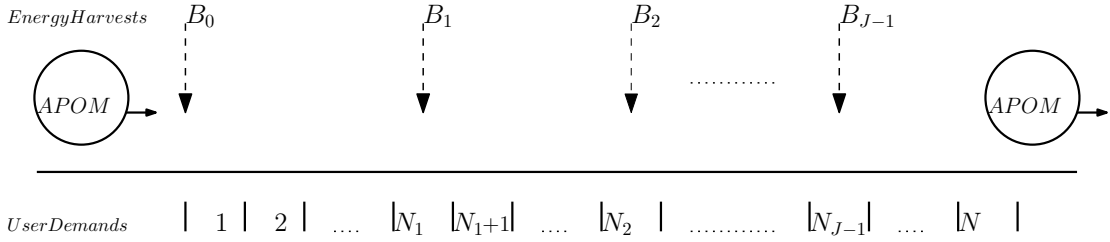}
\caption{Illustration of the problem setup with a sequence of energy arrivals $\{B_0,..., B_{J-1}\}$ and service requests from a total that the APOM encounters on its path.}
\centering
\label{fig:model}
\end{figure*}

In the following sections online policies are proposed for deterministic variations of this problem. One practical motivation for the deterministic formulation is that, for some energy sources, solar energy in particular, the energy harvest profile is quite predictable as it exhibits a daily pattern \cite{ Tekbiyik2013}. In our stochastic problem formulation, energy harvesting is modelled as a random stochastic process using similar assumptions as in recent related literature \cite{Ephremides2012, Tan2014, elif_thesis}. Both deterministic and stochastic approaches for this system setup bring discrete power levels and utilities, which is also consistent with practical concerns. Following knapsack terminology, APOM is characterized by its \textit{capacity}, which corresponds  to the amount of energy stored in its battery. The objective is to collect the maximum value over $N$ users while ensuring that total weight does not exceed the service capacity at any time. If the total capacity was static, the problem would be an online knapsack problem which is known to be NP-hard. The problem with intermittent capacity replenishments is the more difficult online knapsack problem with dynamic incremental capacity. 

Below, we start by explicitly stating the offline problem formulation where a duration of time covering a total of $J\geq 1$ energy replenishments. $N_i$ denotes the arrival time of the $i^{th}$ energy replenishment of value $B_i$, the first one being at time $N_0=0$, with amount $B_0$. Each time slot corresponds to one user arrival, and the number of users observed by time $N_J$ is $N$,  with $B_J=0$. The problem is stated in terms of the decision variables $x_n\in \{0,1\}, n=\{1, 2, \ldots, N\}$'s, which indicate the decision to either serve the $n^{th}$ user ($x_n=1$) or pass it up ($x_n=0$).

\begin{problem}{Deterministic (offline) service policy optimization: }
\mbox{Given} $J \in \mathbb{Z}^+,  0<N_1<N_2\ldots N_{J-1}<N_J=N; \; \{v_i, w_i\}, i=1,2,\ldots N; \; \mbox{$B_0, B_1,\ldots B_{J-1}$} \in \mathbb R^+ \mbox{such that } \sum_{j=0}^{J-1}B_j\leq N, $ choose $ x_n\in\{0,1\}, \forall n\in\{1,..N\} $ to: 
\begin{align}
\noindent \mbox{Max. } &\sum_{n=1}^Nv_nx_n \label{eq:knap1}
\\
\label{eq:c}
\noindent \mbox{s.t. } &\sum_{l=1}^{n}w_lx_l\leq \sum_{j=0}^{j_n}B_j \\
\noindent \forall n, 1\leq n\leq N,  &\mbox{ where } j_n=\arg \max_{j\geq 0}\{N_j\leq n\} \nonumber
\end{align}
\label{prob1}
\end{problem}

Note that  $j_n$ is the time slot index of the last energy arrival before time $n$, hence the constraint (\ref{eq:c}) amounts to satisfying ``energy causality", i.e., at any time, the total energy used cannot exceed the total energy harvested.

\section{Service Policy  of APOM through a Stochastic Knapsack Formulation}
\label{sec:stoc}

Now, we will impose a stochastic structure on the problem in order to obtain an online algorithm via DP. Suppose $(v_n, w_n)\in {\mathcal S}_1 \times {\mathcal S}_2 \subset \mathbb R^+ \times  \mathbb R^+$ such that $|{\mathcal S}_1 \times {\mathcal S}_2|=K<\infty$. That is, there are $K$ different types of users (value-weight pairs). Furthermore, for each slot $n$, a user of type $k$, $1\leq k\leq K$ occurs with probability $p(k)$, independently from all other slots.  Let $\{N_j,~ j=1,..n\}$ and $\{B_j \in [0,M], j=1,.. N\}$  be positive-valued IID random sequences, where $M<\infty$. Due to the finiteness of $M$, the total energy stored under any policy at any time is finite.

The objective is to maximize an expected total reward collected over the time horizon.

\begin{problem}{Stochastic Service policy optimization: \mbox{Given} $J \in \mathbb{Z}^+,  0<N_1<N_2\ldots N_{J-1}<N_J=N; \; \{v_i \in S_1, w_i \in S_2\}, ~p(k)=Prob(v_i=v(k)\mbox{ and }w_i=w(k)),~ i=1,2,\ldots N ,~ k=1,2,\ldots K; \; \mbox{$B_0, B_1,\ldots B_{J-1}$} \in \mathbb [0,M]~ \mbox{such that } \sum_{j=0}^{J-1}B_j\leq N, $}
\label{prob2}

\mbox{Choose a policy} $\pi=[x_1, \ldots, x_n]\in \{0,1\}^N$ \mbox{to}
\begin{align}
\noindent \mbox{Maximize:  } &\displaystyle\mathbb{E}\left(\sum_{n=1}^{N}v_nx_n\right)
\label{eq:knap1_s}
\\
\label{eq:causality}
\noindent \mbox{s.t. }  &\sum_{l=1}^{n}w_lx_l\leq \sum_{j=0}^{j_n}B_j \\
\noindent \forall n, 1\leq n\leq N,  &\mbox{ where } j_n=\arg \max_{j\geq 0}\{N_j\leq n\} \nonumber
\end{align}
\end{problem}

\subsection{Optimal Online Policy via Dynamic Programming}

Note that feasible policies $\pi=[x_1, \ldots, x_n]$ belong to the action space $\{0,1\}^N$, and further satisfy energy causality, given in (\ref{eq:causality}). For every feasible policy, the remaining energy, $e_n$, at the beginning of slot $n$ is given by:
\[
e_n=\sum_{j=0}^{j_n}B_j-\sum_{i=1}^{n-1}x_iw_i
\]
recall that $j_n$ is the time slot index of the last energy arrival up to (and including) slot $n$.

As user and energy arrival sequences are IID, it is straightforward to see that the state at time slot $n$ is captured by $(v_n, w_n, e_n)$, and $e_n$ evolves as a Markov decision process where the action taken on the $n^{th}$ slot,  $x_n$, is:
\begin{align}
x_n(v_n, w_n, e_n)=\left\{
     \begin{array}{lr}
       1 &  (\mbox{\textit{transmit}})\\
       0 &  (\mbox{\textit{defer}})
     \end{array}
   \right. \
\end{align}

For ease of exposition, we make two other simplifying assumptions. First, suppose $|B_j|=1$ for all $j$, equivalently,  let  energy arrivals be given by an IID Bernoulli sequence $\{Q_n \in \{0,1\}, n=1,\ldots,N\}$ over slots. Seconly, we set costs to unity, $w(k)=1$, for all types of users $k$.  Consequently, the feasibility condition reduces to $\sum_{n=1}^j x_n\leq e_0+\sum_{n=1}^{j-1} Q_n$.

Let $V_{\pi}(e_n,k_n)$ be the objective function of this Markov Decision Problem; the expectation of the total value collected from slot $n$ starting with energy level $e_n$ and user type $k_n$ till the end of time horizon $N$ under policy $\pi$. A Dynamic Programming (DP) problem recursion can be written  starting from the last step, namely $N^{th}$ slot:

\begin{align}
V_N^*(e_N,k_N)= v(k_N), \forall e_N \geq 1
\end{align} and go backwards using:
\begin{align}
V_{x_{n}}(e_n,k_n)= &v(k_{n})x_{n}\\+&\displaystyle\mathbb{E}_{(k',Q)}\{V^*_{n+1}(e_{n}-w(k_{n})x_{n}+Q_n,k'_{n+1})\}\nonumber
\end{align}
For a current user $n$,  expected value till the end of the horizon, after choosing to transmit to the current user (i.e. $x_n=1$), is denoted as $V_1(e_n,k_n)$ whereas deferring that user  (i.e. $x_n=0$) is represented as $V_0(e_n,k_n)$. Comparing these quantities, the optimal expected value may be stated as:
\begin{eqnarray}
V^*_n(e_n,k_n)&=& \max_{x_n \in \{0,1\}}V_{x_n}(e_n,k_n)\nonumber\\
&=&\max\{V_1(e_n,k_n),V_0(e_n,k_n)\}
\label{onsol2}
\end{eqnarray}

Backward induction of DP reveals a threshold based policy where the decision maker adopts a conservative attitude at the initial slots and behaves more greedily toward the end of the horizon. To illustrate, the APOM will try to conserve its energy to serve users with higher utility at the beginning. A pseudo code for the dynamic programming solution is given in Algorithm 1.

\begin{algorithm}
\label{algorithm_DP}
\caption{DP solution to the problem for finite horizon}
\footnotesize\begin{algorithmic} 
\FOR {$e = 0$ to $E$}
\FOR {$k = 1$ to $K$}
\STATE $V(e_{N+1},k_{N+1})=0$ \COMMENT{Initialization step}
\ENDFOR
\ENDFOR
\FOR {$n= N$ to $1$}
\FOR {$e = 0$ to $E$}
\FOR {$k = 1$ to $K$}
\IF{$w(k_n)>e_n$}
\STATE $V(e_n,k_n)=E_{(k',Q)}\{V^*_{n+1}(e_n+Q_n,k'_{n+1})\}$
\ELSE
\STATE $V(e_n,k_n)= \max\{\mathbb{E}_{(k',Q)}\{V^*_{n+1}(e_n+Q_n,k'_{n+1})\},    v(k_n)+\mathbb{E}_{(k',Q)}\{V^*(e_n-w(k_n)+Q_n,k'_{n+1}\}\}$ \COMMENT{Recursive equation}
\ENDIF
\RETURN {$V(e_n,k_n)$} 
\ENDFOR
\ENDFOR
\ENDFOR
\end{algorithmic}
\end{algorithm}

\subsection{Structure of the Optimal Policy}

The structure of the optimal policy  may be obtained based on the DP relaxation.

\subsubsection{Existence of threshold}

\begin{lemma}
\label{prop1}
For a given k and n, the state defined as expected total reward $V_{x_{n}}(e_n,k_n)$ is super-modular in available energy and decision pair $(e_n,x_n)$, that is, for any $0\leq e_0 \leq e_1 <\infty$, $V_1(e_1,k_n)+V_0(e_0,k_n)\geq V_0(e_1,k_n)+V_1(e_0,k_n)$ for $1\leq n \leq N$. 
\end{lemma}
\begin{proof}
The stated super-modularity  corresponds to the statement:
\begin{align}
V_1(e_1,k_n)-V_0(e_1,k_n) \geq V_1(e_0,k_n)-V_0(e_0,k_n) \label{eq:holding}
\end{align}
 In our construction energy harvests are in random increments of $1$ hence there is an integer energy level difference between $e_1$ and $e_0$. It suffices to prove this statement for $e_1-e_0=1$ (it can be easily extended to higher differences by iteration of the same argument.) Let $e_1=e_0+1=e+1$. The equality above can be shown to hold for $1\leq n\leq N$ by the following argument: 
\begin{align}
V_1(e_n=e,k_n)-V_0(e_n=e,k_n) = v(k_n)\nonumber\\+\sum_{k'=1}^K p(k')(q(V^*_{n+1}(e,k'_{n+1})-V^*_{n+1}(e+1,k'_{n+1}))\nonumber\\+(1-q)(V^*_{n+1}(e-1,k'_{n+1})-V^*_{n+1}(e,k'_{n+1})))\label{eq:eqproof1}
\end{align} and,
\begin{align}
V_1(e_n={e+1},k_n)-V_0(e_n={e+1},k_n)= v(k_n)\nonumber\\+\sum_{k'=1}^K p(k'_{n+1})q(V^*_{n+1}({e+1},k'_{n+1})-V^*_{n+1}({e+2},k'_{n+1}))\nonumber\\+(1-q)(V^*_{n+1}(e,k'_{n+1})-V^*_{n+1}({e+1},k'_{n+1})))\label{eq:eqproof2}
\end{align} 
By subtracting  (\ref{eq:eqproof1}) from (\ref{eq:eqproof2}), a sufficient condition  for (\ref{eq:holding}) to hold $\forall n \geq 1$ becomes:
 \begin{equation}
V^*_n(e,k_n)-V^*_n(e-1,k_n) \geq V^*_n(e+1,k_n)-V^*_n(e,k_n) \label{eq:eqproofsuff}
\end{equation} 
Then, the condition of (\ref{eq:eqproofsuff}) is proved by induction. First, the condition is satisfied when $n=N$, that is both sides of the equation are equal to $0$. Second, if it is true for some $n+1$ as in (\ref{eq:rec}) then it is proved to hold for $n$.
\begin{align}
V^*_{n+1}(e,k_{n+1})-V^*_{n+1}({e-1},k_{n-1}) \geq  V^*_{n+1}({e+1},k_{n+1})-V^*_{n+1}(e,k_{n+1})\label{eq:rec}
\end{align} 
We will examine the three cases corresponding 3 energy states ($e+1,e,e-1$) and the three optimal decisions ($x_1,x_2,x_3 \in \{0,1\}$) respectively to show whether the inequality given in (\ref{eq:ineq}) holds. 
\begin{align}
V_{x_1}({e+1},k_n)-V_{x_2}(e,k_n)-(V_{x_2}(e,k_n)-V_{x_3}({e-1},k_n))\leq 0 \label{eq:ineq}
\end{align} This is proved by adding $-V_{x_1}(e,k_n)+V_{x_1}(e,k_n)$ and  $-V_{x_3}(e,k_n)+V_{x_3}(e,k_n)$ to the LHS of the inequality such that:
\begin{align}
V_{x_1}({e+1},k_n)-V_{x_1}(e,k_n)+V_{x_1}(e,k_n)\nonumber\\-V_{x_2}(e,k_n)-(V_{x_2}(e,k_n)-V_{x_3}({e-1},k_n))\nonumber\\-V_{x_3}(e,k_n)+V_{x_3}(e,k_n)\leq 0
\label{eq:proof3}
\end{align}

By optimality of $x_2$ for energy state $e$, $V_{x_1}(e,k_n)-V_{x_2}(e,k_n)$ statement is already smaller than or equal to $0$. Same property holds for the $V_{x_3}(e,k_n)-V_{x_2}(e,k_n)$ statement. Therefore, we should only consider the remaining terms. For each possible case of $[x_1, x_3]\in \{0,1\}^2$, the inequality in (\ref{eq:proof3}) is shown to be satisfied. For example, lets examine the case where $x_1=1, x_2=1$:
\begin{align}
V_{1}({e+1},k_n)-V_{1}(e,k_n)-(V_{1}(e,k_n-V_{1}({e-1},k_n))\nonumber \\ =\sum_{k'=1}^K p(k'_{n+1})q(V^*_{n+1}({e+1},k'_{n+1})-V(e,k'_{n+1})\nonumber\\-V^*_{n+1}(e,k'_{n+1})+V^*_{n+1}({e-1},k'_{n+1}))\nonumber\\+(1-q)(V^*_{n+1}(e,k'_{n+1})-V^*_{n+1}({e-1},k'_{n+1})\nonumber \\-V^*_{n+1}({e-1},k'_{n+1})+V^*_{n+1}({e-2},k'_{n+1})\leq 0
\end{align}
The above inequality holds since the difference is assumed to be non increasing in available energy ($e_n$). Similar steps may be followed for all combinations of $x_1$ and $x_3$ where $[x_1, x_3]\in \{0,1\}^2$ and give the same result.
Hence, the total expected reward is a super-modular function in $(e_n,x_n)$.  
\end{proof}

\begin{theorem}
\label{thm:thresh1}
The optimal policy is a threshold type policy in the available energy $e_n$ at each slot $n$ for a user $k_n$, thus there is a threshold $\eta$ defined as:
$x_n(k_n)=\left\{
     \begin{array}{lr}
       1 & : e_n \geq \eta_n(k_n)\\
       0 & : e_n < \eta_n(k_n)
     \end{array}
   \right. \label{eq:threshold_function}$
\end{theorem}

\begin{proof}
Let $\{e_1, e_2, e_3\}$ be the available energies at three decision instants such that $e_1<e_2<e_3$. Suppose there exists an optimal policy which chooses to transmit at energy levels $e_1$ and $e_3$  while deferring the user at the energy level $e_2$. This contradicts Lemma \ref{prop1}. Therefore, the crossover from \textit{Defer} to \textit{Transmit} happens only once as $e_n$ is increased (holding all other parameters constant), i.e. there is a threshold.
\end{proof}

\subsubsection{Monotonicity of threshold}

\begin{lemma}
\label{lemmamon1}
Expected total reward $V_{x_n}(e_n,k_n)$ is super-modular in slot index and decision pair ($n,x_n$), that is $V_1(e_{n+1},k_{n+1})+V_0(e_n,k_n)\geq V_0(e_{n+1},k_{n+1})+V_1(e_n,k_n)$. 
\end{lemma}
\begin{proof}
Following similar steps as in the proof of Lemma \ref{prop1}, supermodularity corresponds to the statement:
\begin{align}
\label{eq:supermod}
V_1(e_{n+1},k_{n+1})-V_0(e_{n+1},k_{n+1})\nonumber\\\geq V_1(e_n,k_n)- V_0(e_n,k_n)
\end{align}
\end{proof}

\begin{corollary}
The threshold function on the available energy to serve a user, $\eta_n(k_n)$ defined in Theorem \ref{eq:threshold_function} is monotonically non-increasing with slot number $n$. 
\label{prop4}
\end{corollary}
\begin{proof}
Let $n\geq 1$ be the first slot index such that the threshold increases from $n$ to $n+1$. This means the policy chooses to transmit to a user of some type $k$ at $n$ while deferring a user of the same type at slot $n+1$, for the same starting energy $e_n$. By  (\ref{eq:supermod}) this policy can be improved by reversing this decision hence cannot be optimal. 
\end{proof}

From the analysis above, a monotonic threshold type policy is shown to be optimal. However,  the computation of the exact optimal threshold without loss of generality requires extensive calculations and is not the focus of this study. Instead, suboptimal threshold policies are proposed by exploiting the structure of the optimal policy.

\subsection{Suboptimal Solution: Expected Threshold Policy}

DP provides an optimal solution for 0/1 dynamic and stochastic knapsack problem with  growing capacity; however its computational complexity increases exponentially with $N$, which is consistent with the NP-hardness of the problem \cite{Bertsekas}. In this section, a computationally cost- effective suboptimal policy called \textit{Expected Threshold Policy} \cite{Tan2013,Tan2014} will be adopted for this problem.

First, we define the following bound on the expectation  of energy depletion (RHS of (\ref{eq:bound})) at slot $n$ if the available energy is $e_n$ and expected harvest amount from slot $n$ till the end of time horizon $N$  is denoted as $\displaystyle\sum_{m=n}^{N-1}\mathbb{E}\{Q_m|Q_1^{n-1}\}$ .
\begin{align}
e_n+\displaystyle\sum_{m=n}^{N-1}\mathbb{E}\{Q_m|Q_1^{n-1}\} \geq \displaystyle\sum_{m=n+1}^{N}\mathbb{E}\{w_mx_m\}
\label{eq:bound}
\end{align}
After stating a bound on the expected energy consumption from slot $n$ till the end of time horizon in  (\ref{eq:bound}), a computationally cost effective suboptimal policy called ``Expected Threshold" is proposed  in (\ref{eq:expthr1}) as follows:
\begin{align}
x_n(k_n,e_n)=\left\{
     \begin{array}{lr}
       1 & : e_n \geq \widehat{\eta}_n\\
       0 & : e_n < \widehat{\eta}_n
     \end{array}
   \right.
   \label{eq:expthr11}
\end{align}
where 
\begin{align}
\widehat{\eta}_n = \displaystyle\sum_{m=n+1}^{N}\mathbb{E}\{w_mx_m\}-\displaystyle\sum_{m=n}^{N-1}\mathbb{E}\{Q_m|Q_1^{n-1}\}
\label{eq:expthr21}
\end{align}
As  it is seen from (\ref{eq:expthr1}) and (\ref{eq:expthr2}), APOM makes a decision to serve a user of type $k$ appearing in slot $n$ ($k_n$) if the available energy $e$ at slot $n$ ($e_n$) is greater than or equal to the threshold level $\eta_n$. $\eta_n$ is stated as the difference between the expected energy consumption for users with higher value and expected energy replenishment from slot $n$ till the end of horizon $N$.

As an example, if the weights are all equal to one and harvest process is IID, then 
\begin{align*}
\displaystyle\sum_{m=n+1}^{N}\mathbb{E}\{w_mx_m\}= \displaystyle\sum_{m=n+1}^{N}{p(k_{n})w(k_{n})}=(N-n+1)\displaystyle\sum_{k'=k_n+1}^K{p_{k'}}  \nonumber\\\displaystyle\sum_{m=n}^{N-1}\mathbb{E}\{Q_m|Q_1^{n-1}\}=(N-n+1)q
\end{align*}
so that the expected threshold policy becomes:
\begin{align}
x_n(k_n)=\left\{
     \begin{array}{lr}
       1 & : e_n \geq \eta_n(k_n)\\
       0 & : e_n < \eta_n(k_n)
     \end{array}
   \right.
   \label{eq:expthr1}
\end{align}
where 
\begin{align}
\eta_n(k_n) = (N-n+1)(\sum_{k'=k_n+1}^K{p_{k'}}-q)
\label{eq:expthr2}
\end{align} $\forall n=\{1,...,N\}$ and users of type $k$ at each slot $n$, $k_n\in \{1,...,K\}$ are arranged such that priority of a user ($v_n/w_n$) increases with increasing $k_n$.  

To examine the performance of the expected threshold policy in Section (\ref{sec:results}), two more benchmark heuristics are also defined as follows: 
\begin{definition}
\label{def:greedy}
Greedy policy is a policy that serves an encountered user whenever there is available energy to serve it; that is $x_n=1$ iff $e_n\geq w_n$.
\end{definition} 
\begin{definition}
\label{def:cons}
Conservative policy is a policy that serves only the best class of user in terms of $value/weight$ when there is available energy to serve it, i.e. $x_n=1$ iff $v_n/w_n=max_{k=1,\ldots,K}{v_k/w_k}$ and $e_n\geq w_n$.
\end{definition}

\subsection{A Performance Upperbound }

An upper-bound performance analysis can be conducted considering the total expected reward (value) $V_N$ for two user type case (i.e., $(v_1,w_1), (v_2,w_2)$) over a finite horizon $N$. $V_N$ can be written as:
\begin{align}
V_N=\displaystyle\sum_{n=1}^{N}{v_1x_n(1)+v_2x_n(2)}\label{eq:total_value}
\end{align} where $x_n(k_n)\in \{0,1\}$ is the decision whether to serve user type  $k_n\in\{1,2\}$ or defer it, such that $v_1/w_1 \geq v_2/w_2$. Also note that $x_n(1)+x_n(2)\leq 1$, $\forall n$. 
Total energy consumption rate should be no greater than the replenishment rate, therefore,
\begin{align}
\displaystyle\sum_{n=1}^{N}{w_1x_n(1)+w_2x_n(2)}\leq \displaystyle\sum_{n=0}^{N-1}q
\end{align} This inequality can be rewritten as:
\begin{align}
\displaystyle\sum_{n=1}^{N}{x_n(2)}\leq \displaystyle\sum_{n=0}^{N-1}q/w_2-\displaystyle\sum_{n=1}^{N}w_1/w_2x_n(1)\label{eq:xn2}
\end{align} After substituting (\ref{eq:xn2}) into (\ref{eq:total_value}):
\begin{align}
V_N\leq (v_1-v_2w_1/w_2)\displaystyle\sum_{n=1}^{N}{x_n(1)}+ \displaystyle\sum_{n=0}^{N-1}v_2q/w_2
\end{align}
Since the average amount of energy consumed by serving \textit{user type 1} is limited by either the average amount of energy harvested till the end of horizon ($\displaystyle\sum_{n=0}^{N-1}q$) or the average amount of energy requested  ($\displaystyle\sum_{n=1}^{N}p_1w_1$), the following bound is attained: 
\begin{align}
\displaystyle\sum_{n=1}^{N}x_n(1)\leq min\{\displaystyle\sum_{n=0}^{N-1}q/w_1,\displaystyle\sum_{n=1}^{N}p_1\}
\end{align} 
Therefore, the total expected reward has a performance upper-bound as follows:

\begin{align}
V_N \leq \left\{
     \begin{array}{lr}
       \frac{v_1qN}{w_1} & : min\{p_1,\frac{q}{w_1}\}=\frac{q}{w_1} \\
       \frac{v_2qN}{w_2}+(v_1-\frac{v_2w_1}{w_2})p_1 & : min\{p_1,\frac{q}{w_1}\}=p_1 
     \end{array}
   \right.
\end{align}

\section{Service Policy of APOM through a Knapsack Formulation over Deterministic Model}

\label{sec:deter}

APOM has to adopt an efficient and fast decision making strategy as a new user demand appears. In such problems, if a well defined threshold could be stated, then the threshold based  decision mechanism gives a satisfactory result in terms of overall performance and computational complexity. Hence, we shall mainly look for threshold based schemes which provably exhibit experimentally strong performance.

\subsection{An Online Policy with Deterministic Threshold Method}

\label{det_a}

Threshold based decision rules are examined in this section, where values and weights of the encountered users are compared with a time-varying threshold. In addition to time, the threshold may also be a function of the fraction of remaining capacity in the battery. To consider the deterministic online knapsack problem in a threshold based scheme, upper and lower bounds on the user rate, energy requirement and energy harvesting will be assumed, which are not unrealistic considering practical correspondents to these limitations exist.

The cost efficiency, \textit{value/weight}   ($v/w$), will be the critical decision metric for each user. The instantaneous threshold is defined as a monotonic increasing function of $z_n=\frac{\sum_{m=1}^nx_mw_m}{B}$, the fraction of the capacity used up by the $n^{th}$ slot. Following \cite{Zhou2008}, where an optimal threshold scheme was developed for the static capacity 0/1 KP, we restrict attention to the case where the value/weight values are upper and lower bounded by $U, L > 0$, i.e. $L\leq \frac{v}{w}\leq U$, and define the threshold function as\footnote{The form of this threshold function is found through linear programming and shown to achieve an optimal competitive ratio in \cite{Zhou2008}}:
\begin{align}
\Psi(z) = (\frac{Ue}{L})^z\frac{L}{e} \hspace{0.1in} where \hspace{0.1in} L\leq \frac{v}{w}\leq U
\end{align}where $\textbf{e}$ denotes the natural logarithm. At each slot $n$, the value/weight value of the upcoming user is compared with the threshold $\Psi(z_n)$.  The threshold based decision rule is the following:

\noindent Accept user $n$ provided it does not violate the current remaining knapsack capacity and $v_n/w_n\geq \Psi(z_n)$.

For a static KP, $z_n$ thus the threshold is monotone nondecreasing, which corresponds to being more willing to include users early on, and being very selective as $z_n$ increases toward 1. In our problem, the knapsack capacity is not static but gets incremented at arbitrary instants, at arbitrary amounts. The detailed extension of the threshold function with the optimal competitive ratio is revealed in Section \ref{append} for the dynamic capacity knapsack problem provided some  preconditions. Extending the above threshold to the two extreme cases where,
\begin{itemize}
  \item Complete information about the incremental amounts is available.
  \item No information about the increment amounts is available.
\end{itemize}
  the fraction $z_n$ may be defined in two different ways.

\begin{definition} Monotone Threshold. Define $z_{\rm{mon},n}=\frac{\sum_{m=1}^{n}x_mw_m}{B}$, where $B$ is the total amount of energy $B=B_0+B_1+...+B_{J-1}$ collected from all harvests.  
\label{def_mon}
\end{definition}

As an alternative way to the monotone threshold approach, we define a ``Jumping Threshold" as a piecewise monotone function of the current fraction in each energy harvest interval. It utilizes the the amount of energy harvested up to that time instant as the denominator of the fraction.

\begin{definition}
Jumping Threshold. For each $n$, let $J(n)$ be the time of the last harvest before time $n$. The fraction  of filled capacity at time $n$ is defined as $z_{\rm{jump},n}=\frac{\sum_{m=1}^{n}x_mw_m}{B_0+\ldots+B_{J(n)}}$.
\end{definition}

Clearly, this second threshold function is monotone nondecreasing between harvest instants, and jumps down whenever a new harvest occurs. As opposed to the latter threshold function which assumes prior knowledge of all harvest amounts over the problem horizon, this one is an online algorithm by construction.

A common success metric for a deterministic online algorithm is its \textit{competitive ratio}, the worst-case ratio of the algorithm's performance to the optimal offline solution under the same input.  An online algorithm $A$  for a user sequence $\gamma$ that is $\alpha$-competitive satisfies the following:
\begin{align}
\frac{OPT(\gamma)}{A(\gamma)}\leq\alpha, \hspace{0.1in} {\rm{where}} \hspace{0.1in} \alpha \geq 1
\end{align}
where $OPT(\gamma)$ and $A(\gamma)$ are the values obtained from optimal offline algorithm and the proposed online heuristic $A$ respectively. Having complete uncertainty in the input, the heuristic proposed should build solutions with a competitive ratio better than the worst-case ratio by $\alpha$. 

\begin{remark}
Under the condition $\sum_{m=1}^{N_1}x_mw_m\leq B_0$, Monotone Threshold guarantees a competitive ratio no more than $\ln(U/L)+1$ assuming two energy harvest intervals, i.e. $J=2$.
\label{remark}
\end{remark}

If the total amount of energy collected from energy harvests
$B=B_0+B_1$  is considered as a static capacity while computing the fraction $z_{\rm{mon},n}$ as in Definition \ref{def_mon}, then the competitive ratio of Monotone Threshold corresponds to the the same constant competitive ratio which is optimal in the case of  online KP with static and presumably high capacity.   The proof for Remark \ref{remark} is given in Section \ref{append} since it is mainly an extension of the proof for online KP with static capacity in \cite{Zhou2008} to the dynamic capacity case. Next, we will propose different threshold generation methods using different optimization tools, namely genetic algorithms and fuzzy logic. 


\subsection{Threshold method via Genetic Optimization}

Genetic Algorithm (GA) is a widely used search heuristic, also called a metaheuristic, that uses the process of natural selection as a model. In the computational science, engineering, bioinformatics,  economics,  manufacturing, mathematics, physics and many other fields,  this heuristic is utilized for  optimization purposes to various problems. GA is a widely applied technique for optimization and search problems, especially NP-hard ones including KPs. Basically, candidate solutions are stochastically selected, recombined, mutated, either eliminated or retained based on relative fitness; even when the original problem is based upon a deterministic model.

We propose the implementation of this stochastic approach to our deterministic problem with the twist that the knapsack capacity may also change as  solutions evolve toward better ones in time. Thus, generation adaptation and the capacity incrementation need to be jointly taken into account. To apply GA on a fraction based scheme, a chromosome is chosen as a vector that defines a threshold for each region of values the fraction may take. For this purpose, the values that  remaining fraction of capacity ($z$) can take are quantized in the following manner: The range of fraction ($[0,1]$) is divided into equal regions as $[t_1  t_2 ... t_{1000}]$, where $t_i$ corresponds to the threshold for region $i$. The threshold as the outcome is optimized over a randomly generated expected  user sequences as follows:

\begin{align}
\psi(z)=t_i, ~\mbox{where}~z\in [\frac{i-1}{1000}, \frac{i}{1000}].  
\end{align}

A quantization over $1000$ interval is quite sufficient, providing an opportunity to sweep over a wide range. A number of chromosomes are randomly generated at the beginning and their corresponding competitive ratios are found through the fitness function evaluation. The fitness function checks the energy constraint on the available capacity at each step. In addition, capacity is updated at each energy replenishment, so is the fraction $z$.  

The threshold method via genetic optimization provides a constant optimized threshold over the randomly generated user pools and the real time user demands are served using this threshold. The observations on the fraction based method on the natural selection of the best users over generations produce a certain competitive ratio in the best and the worst cases for randomly generated parent sequences, provided and discussed in Section \ref{sec:results}. 

In genetic algorithm, the solution approaches to global optimum as long as  proposed algorithm operates over the actual user profile. Since apriori knowledge on the user sequences is not available,  threshold via GA prevails the result for a expected user profile over a precise time interval. To increase the efficiency of the decision mechanism, an adaptive threshold policy is proposed using rule based optimization in the following section.

\subsection{Threshold method via Rule Based Optimization}

\label{sec:rule}
A connected set of well defined rules, consisting of related variables in both the propositions and consequences, can handle uncertain knowledge successfully in decision problems. Although rule based approaches have been implemented in quite a few resource allocation problems in the literature \cite{Alkesh2011}, we have come across no previous studies on the threshold determination via this method. 
\begin{table}
\centering
\caption {Membership Rules of 5 Degrees for Threshold Determination Belonging to the Membership Functions}
\label{tab:rule}
\begin{footnotesize}
    \begin{tabular}{|l|l|l|}
     \hline
    Energy Harvest Closeness & Capacity Fullness & Threshold \\ \hline
    Very-Near         & Very-High                 & Med       \\ \hline
    Very-Near         & High                   & Low       \\ \hline
    Very-Near         & Med                    & Low       \\ \hline
   Very-Near         & Low                    & Very-Low     \\ \hline
    Very-Near         & Very-Low                  & Very-Low     \\ \hline
    Near           & Very-High                 & High      \\ \hline
    Near           & High                   & Med       \\ \hline
    Near           & Med                    & Low       \\ \hline
    Near           & Low                    & Very-Low     \\ \hline
    Near           & Very-Low                  & Very-Low     \\ \hline
    Med            & Very-High                 & High      \\ \hline
    Med            & High                   & Med       \\ \hline
    Med            & Med                    & Med       \\ \hline
    Med            & Low                    & Low       \\ \hline
    Med            & Very-Low                  & Very-Low     \\ \hline
    Far            & Very-High                 & Very-High    \\ \hline
    Far            & High                   & High      \\ \hline
    Far            & Med                    & High      \\ \hline
    Far            & Low                    & Low       \\ \hline
    Far            & Very-Low                  & Low       \\ \hline
    Very-Far          & Very-High                 & Very-High    \\ \hline
    Very-Far          & High                   & Very-High    \\ \hline
    Very-Far          & Med                    & High      \\ \hline
    Very-Far          & Low                    & Med       \\ \hline
    Very-Far          & Very-Low                  & Low       \\ \hline
    \end{tabular}
    
\end{footnotesize}
\end{table}

There are two input memberships functions (MF) assigned to define the decision strategy in each possible case for APOM. Both of the input MFs are defined as trapezoidals of 5 degrees. The output MF is assigned as the desired change in the threshold, the ultimate trend of which will  be used to determine which users to serve eventually. One of the input membership functions is chosen as the closeness to energy harvest instants in terms of the number of user arrivals. This parameter is prominent in real life scenarios since expecting an energy harvest sooner or at a far instant may completely alter action to be taken at that slot. Once, the harvest instant gets closer and closer, the service provider should adopt a greedy attitude since it would serve as  long as its service capacity allows it to do. This metric is chosen to vary between $[0,1]$ where the values closer to 1 denotes that an energy arrival is presumed to happen soon, presented as \textit{Very-Near}. Similarly, \textit{Very-Far} stands for the user arrivals at the beginning of an energy harvest interval where the input MF is set to be in the vicinity of $0$. In addition to the energy replenishment rate, the fraction of the utilized energy of available capacity is a critical measure as well. Thus, the second MF is assigned as the depletion of available energy of APOM. The values vary between $[0,1]$ interval same as the first MF function, ranging from \textit{Very-Low} to \textit{Very-High} in 5 levels.

The behaviour of the threshold function is determined  as shown in Figure \ref{fig:fuzzy1} following the well calibrated rules by using the input MFs in Figures \ref{fig:harvest} and \ref{fig:frac} and the output MF in Figure \ref{fig:thre}. Once the rules are tested, the overall performance is increased using  calibrated MF parameters effectively via \textit{fine-tuning}. 

It should also be noted that the improved performance of this heuristic is largely related with the enlarged problem dimension. The accuracy of decisions leads to an improved utility maximization performance through proposing a 3D solution to a 2D problem, obviously at an increased complexity.
\begin{figure}
\centering
\includegraphics[scale=0.26]{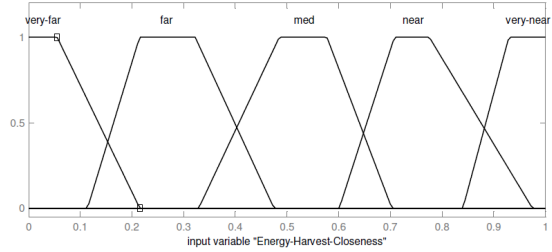}
\caption{Input membership function of energy harvest closeness}
\label{fig:harvest}
\end{figure}
\begin{figure}
\centering
\includegraphics[scale=0.26]{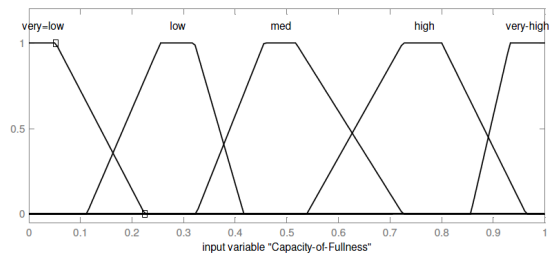}
\caption{Input membership function of fraction of capacity fullness}
\label{fig:frac}
\end{figure}
\begin{figure}
\centering
\includegraphics[scale=0.22]{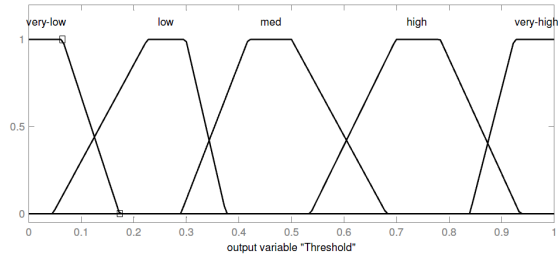}
\caption{Output membership function of threshold}
\label{fig:thre}
\end{figure}
\begin{figure*}
\centering
\includegraphics[scale=0.8]{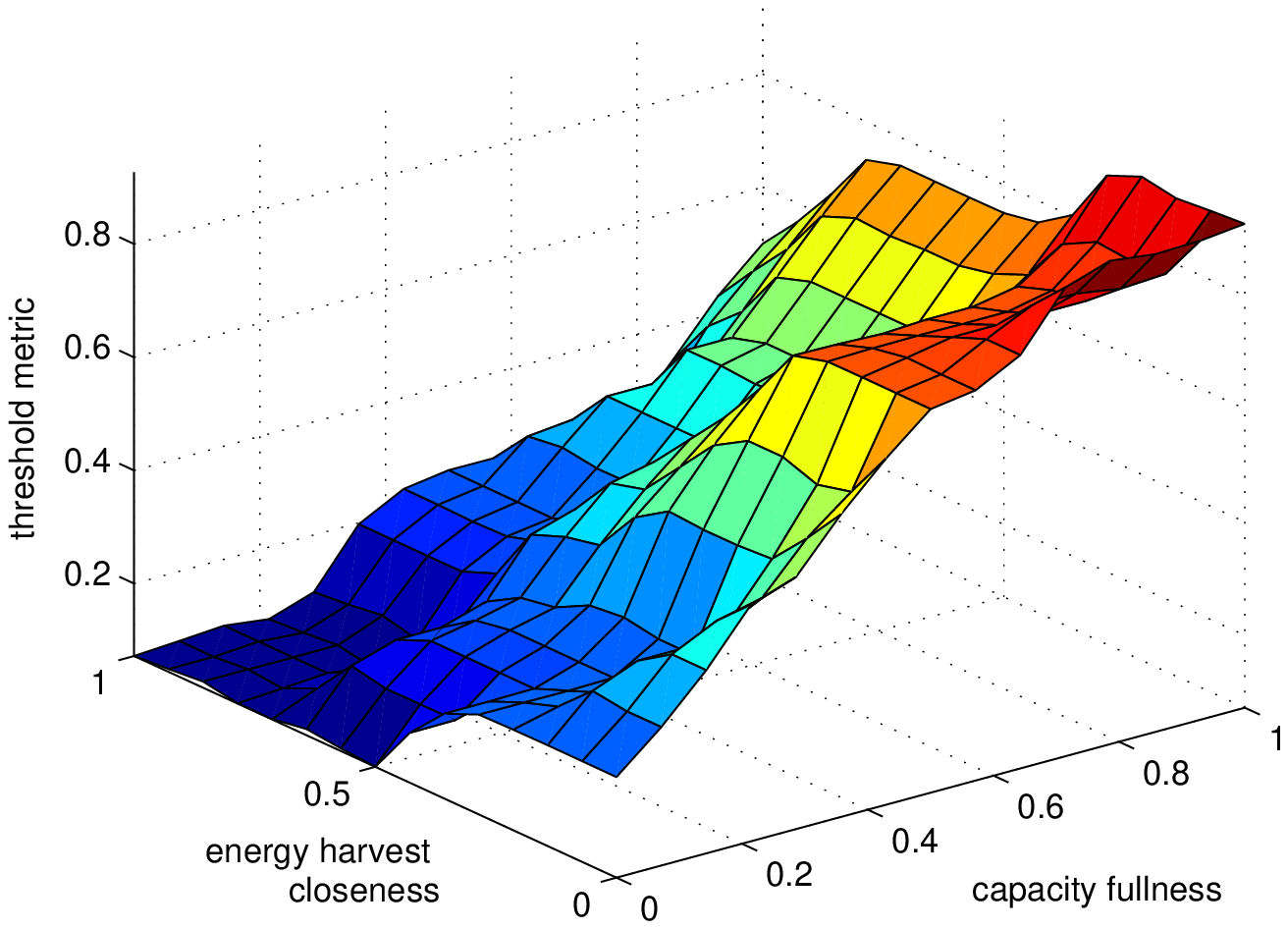}
\caption{Surface graph of threshold function attained via Rule Based Algorithm described in Section  \ref{sec:rule}  indicating the behaviour of output threshold function with respect to input membership functions}
\label{fig:fuzzy1}
\end{figure*}

Five degree MFs for \textit{energy harvest closeness} and \textit{capacity fullness} provide sufficient performance as the decision metric for APOM. The overall threshold system complexity and decision efficiency could be increased via MFs of 7 degree rules. 

In Section \ref{deterresults}, different threshold-based admission mechanisms are investigated and compared on their performance regarding the total utility (rate) they provide. Competitive ratio analysis is used to test the performance of the online algorithms over a deterministic model.

\section{Numerical and Simulation Results}

\label{sec:results}

\subsection{Stochastic Service Policy Optimization Related Results}

Both optimal and suboptimal policies behave more conservative at the beginning, and become more greedy through the end of time horizon. As illustrated in Figures \ref{fig:comp3} and \ref{fig:comp4}, the Expected Threshold Policy performs very close to optimal. Drawbacks of purely greedy and conservative policies are also evident on the figures. Figure \ref{fig:comp3} shows that as the efficient users appear with higher probability, conservative policies outperform the greedy ones considering the expected total utility. On the other hand, when the inefficient users appear with higher probability, greedy policies are more advantageous  than the conservative approach. However, Expected Threshold Policy proposed in this paper is robust against the variations in user distributions. 
\begin{figure*}
\begin{psfrags}
\centering
\includegraphics[scale=0.7]{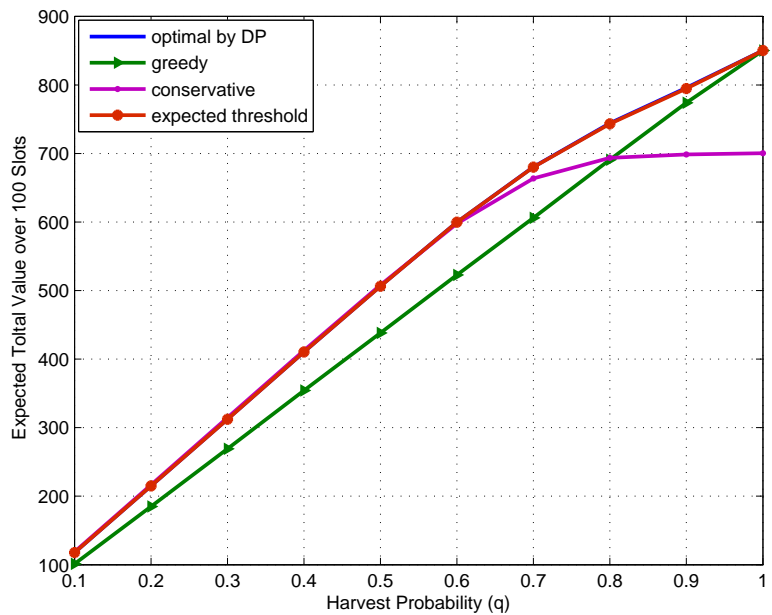}
\end{psfrags}
\centering
\caption{The comparison of the performances for Expected Energy Threshold Policy, Greedy Policy and Conservative Policy wrt. optimal policy when available energy=5, $N=100$, $K=2$ for two different user types with efficiency ratios 10 and 5 (best users appear with high probability e.g. 0.7) }
\label{fig:comp3}
\end{figure*}
\begin{figure*}
\centering
\begin{psfrags}
\includegraphics[scale=0.7]{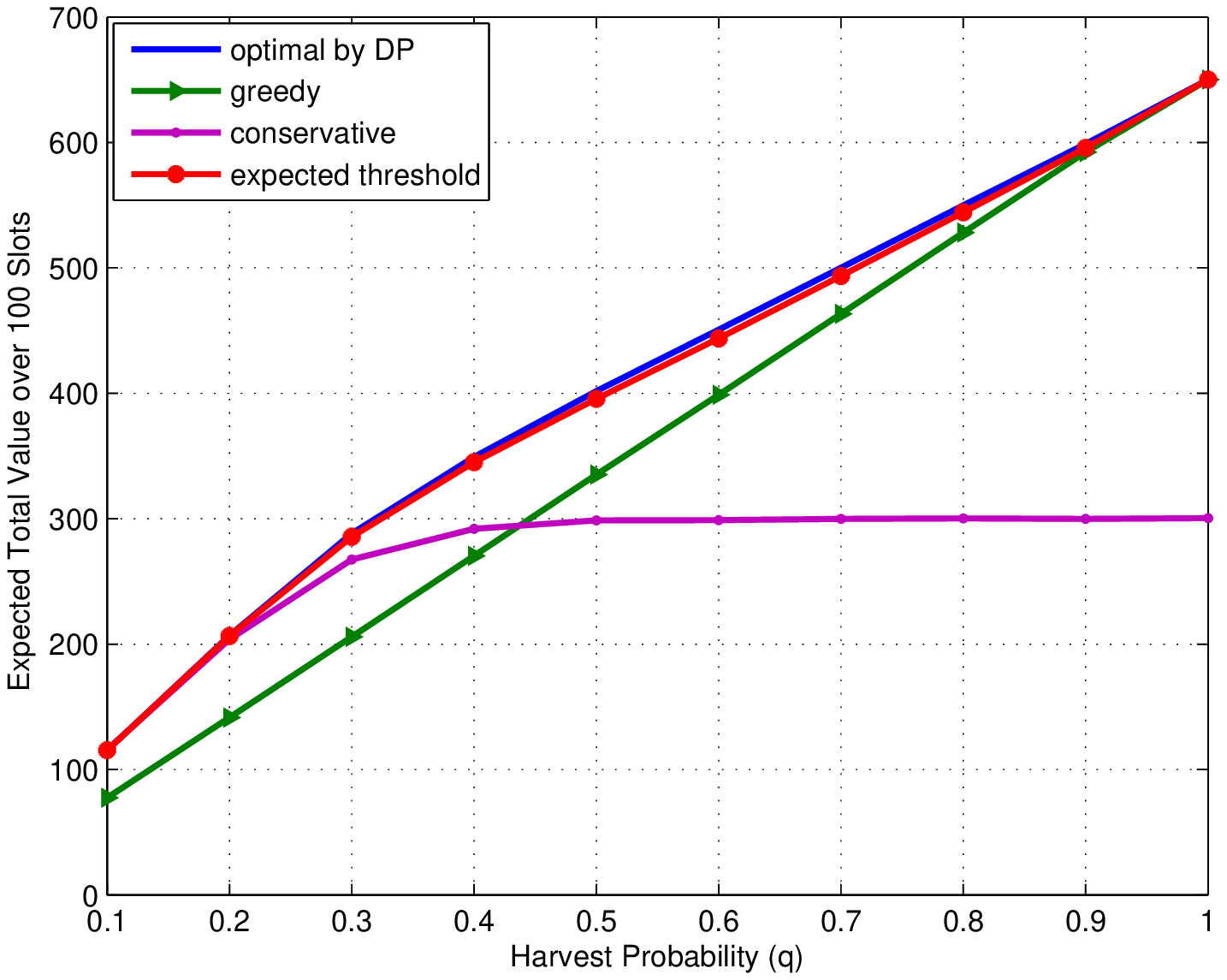}
\end{psfrags}
\centering
\caption{The comparison of the performances for Expected Energy Threshold Policy, Greedy Policy and Conservative Policy wrt. optimal policy when available energy=5, $N=100$, $K=2$ for two different user types with efficiency ratios 10 and 5 (worst users appear with high probability e.g. 0.7) }
\label{fig:comp4}
\end{figure*}
\begin{figure*}
\centering
\begin{psfrags}
\includegraphics[scale=0.8]{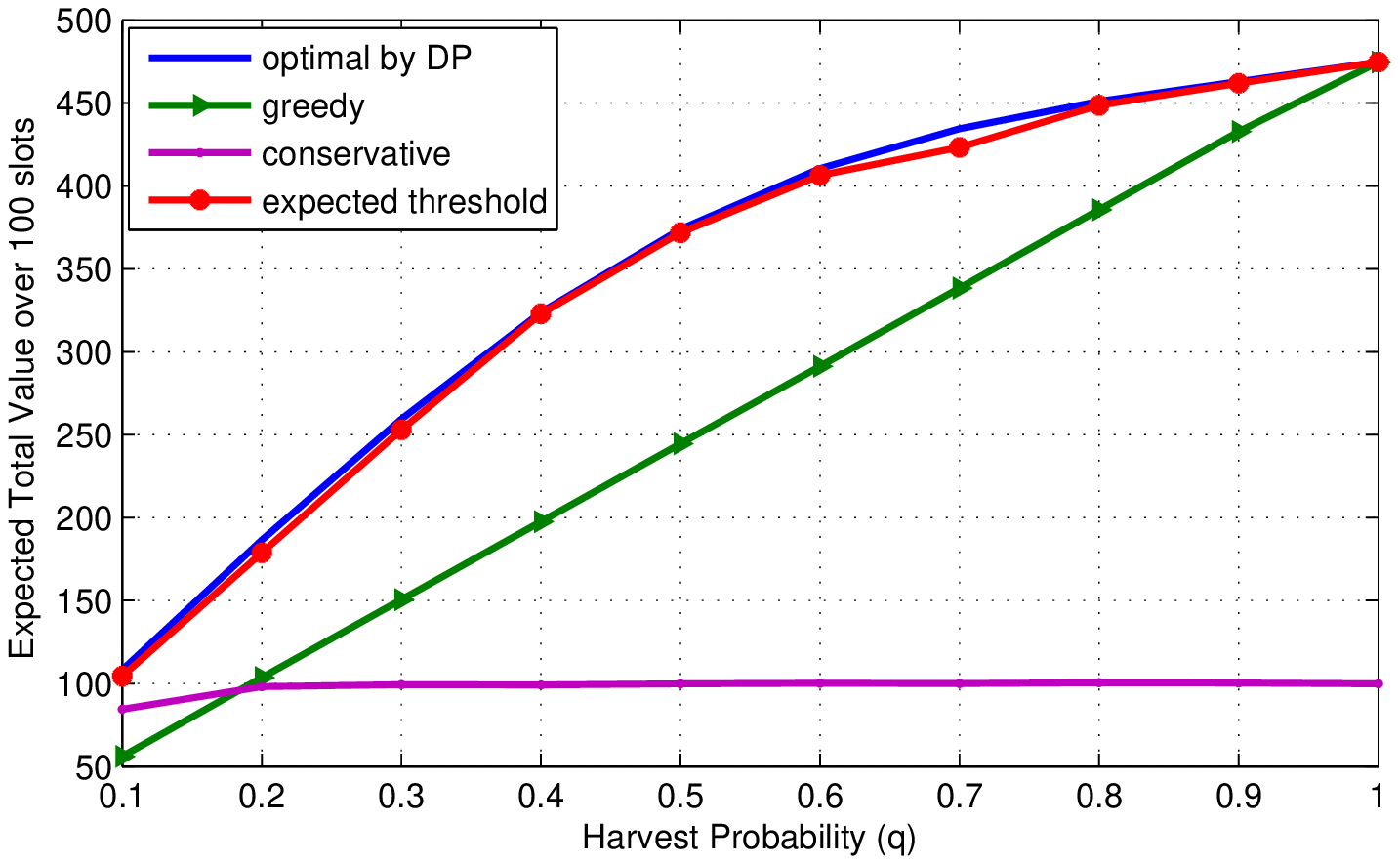}
\end{psfrags}
\centering
\caption{The comparison of the performances for Expected Energy Threshold Policy, Greedy Policy and Conservative Policy with respect to optimal policy when available energy=5 at the beginning, $N=100$, $K=5$ for five different user types with equal weights }
\label{fig:five_user}
\end{figure*}

In order to investigate more general scenarios, number of  user types ($K$) is increased to 5 and simulations have been conducted on the users with both equal weights (energy demand)  as in Figure \ref{fig:five_user} and  different weights as in Figure \ref{fig:five_user_diff}. In Figure \ref{fig:five_user_diff}, simulations have been conducted for five different user types with cost efficiency ($v/w$) ratios given as [10/1, 5/1, 8/4, 5/8, 2/6 1/5] appearing with probabilities [0.3, 0.15, 0.15, 0.3, 0.1]. As it can be seen from Figures \ref{fig:five_user} and \ref{fig:five_user_diff}, Expected Threshold Policy performs quite close to the optimal.

Considering the equal weight scenario, the greedy policy performance is much higher than the conservative one since the unity weight cost does not put a strain on the capacity constraints. However,  the results in Figure \ref{fig:five_user_diff} reveals that the adaptive expected threshold policy outperforms the conservative and greedy policies over user sequences with random weights and $v/w$ cost efficiency ratios.
\begin{figure*}
\centering
\begin{psfrags}
\includegraphics[scale=0.8]{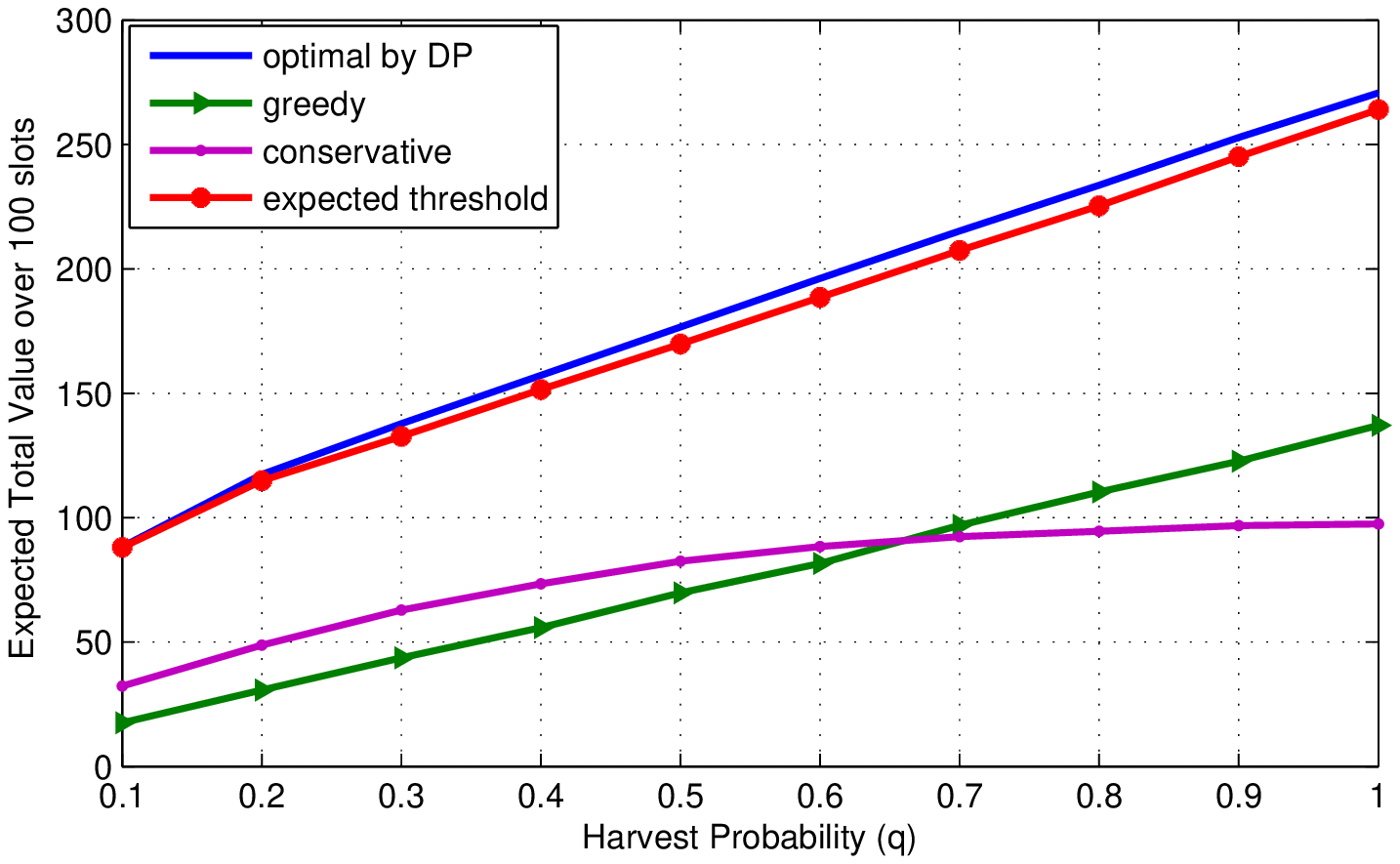}
\end{psfrags}
\centering
\caption{The comparison of the performances for Expected Energy Threshold Policy, Greedy Policy and Conservative Policy with respect to optimal policy when available energy=5 at the beginning, $N=100$, $K=5$ for five different user types with different efficiency (value/weight) ratios and different weights }
\label{fig:five_user_diff}
\end{figure*}

\subsection{Deterministic Service Policy Optimization Related Results}
\label{deterresults}
In this section, competitive ratios of the various policies analysed in the previous sections will studied. As a benchmark, the offline optimal policy will be used, hence the values obtained will be overestimating the competitive ratio with respect to the online optimal for each case. In other words, the competitive ratios of the algorithms may actually perform better than the estimates found here. 

The simulation results shown in Tables \ref{tab:1} and \ref{tab:2} are obtained for the case of predicted  energy harvests. User efficiency ratios take uniformly distributed random values within the interval $6=L\leq \frac{v}{w}\leq U=10$.  The knapsack capacity (energy available at the start of the horizon) is taken as  $2000mJ$ considering one cell for energy harvesting. $1000$ Monte Carlo trials are conducted, generating $1000$ user arrivals on each trial. Results illustrated in Tables \ref{tab:1} and \ref{tab:2} show that even the worst-case competitive ratio never exceeds $1.75$. Moreover, the results for the monotone threshold function are consistent with the worst possible competitive ratio stated in Section \ref{det_a}. Among the tested algorithms, the rule based threshold method has the strongest performance, achieving the lowest worst case competitive ratio.

Next, the energy harvest patterns are considered in a more realistic scenario where the overall resource allocation problem is examined over $10$ energy harvests, assumed to occur in a 24-hour cycle. Also, distinct amounts of the harvests are assumed in this case and assigned arbitrarily to model the potential weather condition changes and different locations of the APOM. The competitive ratio analysis for all of the threshold function methods proposed above yield the results shown in the performance graph \ref{fig:comp1}. The worst-case results illustrated in Figure \ref{fig:comp1} reveal that the Monotone Threshold function and Rule Based Threshold function present closer performance to each other as the user value/weight characteristics are more diverse. However, the rule based threshold provides the best competitive ratio for a less distinct user set as the user diversity ratio ($U/L$) approaches to $0.9$ in the worst-case analysis.

For the average case performance, outputs given in Figure \ref{fig:comp2}, the rule based, monotone and jumping thresholds ensure a similar competitive ratio but the rule based threshold approaches to the optimal solution as the analysis is conducted over a user sequence of similar characteristics.
\begin{table*}
\centering
\caption {Comparison of competitive ratios for different threshold generation methods for 1000 users and  capacity=2000mJ }
    \begin{tabular}{|l|l|l|l|}
    \hline
    Threshold Method & Average competitive ratio &  Worst competitive ratio  & Best competitive ratio \\ \hline
    Monotone threshold & 1.1084              & 1.3100           & 1.0640             \\ \hline
    Jumping threshold            & 1.3700              & 1.7200           & 1.3500            \\ \hline
    GA based threshold                 & 1.1422              & 1.5102           & 1.1087            \\ \hline
    Rule based threshold                              & 1.0362              & 1.2066           & 1.0229            \\ \hline
    \end{tabular}
    \label{tab:1}
\end{table*}
\begin{table*}
\centering
 \caption {Comparison of performances for different threshold generation methods with optimal offline solution for 1000 users and  capacity=2000mJ }
    \begin{tabular}{|l|l|l|l|}
    \hline
    Threshold Method & Average total value & Worst total value & Best total value \\ \hline
    Offline optimal soln.                          & 17599               & 17167             & 18050            \\ \hline
    Monotone threshold & 15880               & 13374             & 16647            \\ \hline
    Jumping threshold            & 12778               & 10221             & 13103            \\ \hline
    GA based threshold                 & 15416               & 11581             & 16042            \\ \hline
    Rule based threshold                              & 17003               & 14524             & 17163            \\ \hline
    \end{tabular}
    \label{tab:2}
\end{table*}
\begin{figure*}
\centering
\begin{psfrags}
\includegraphics[scale=0.4]{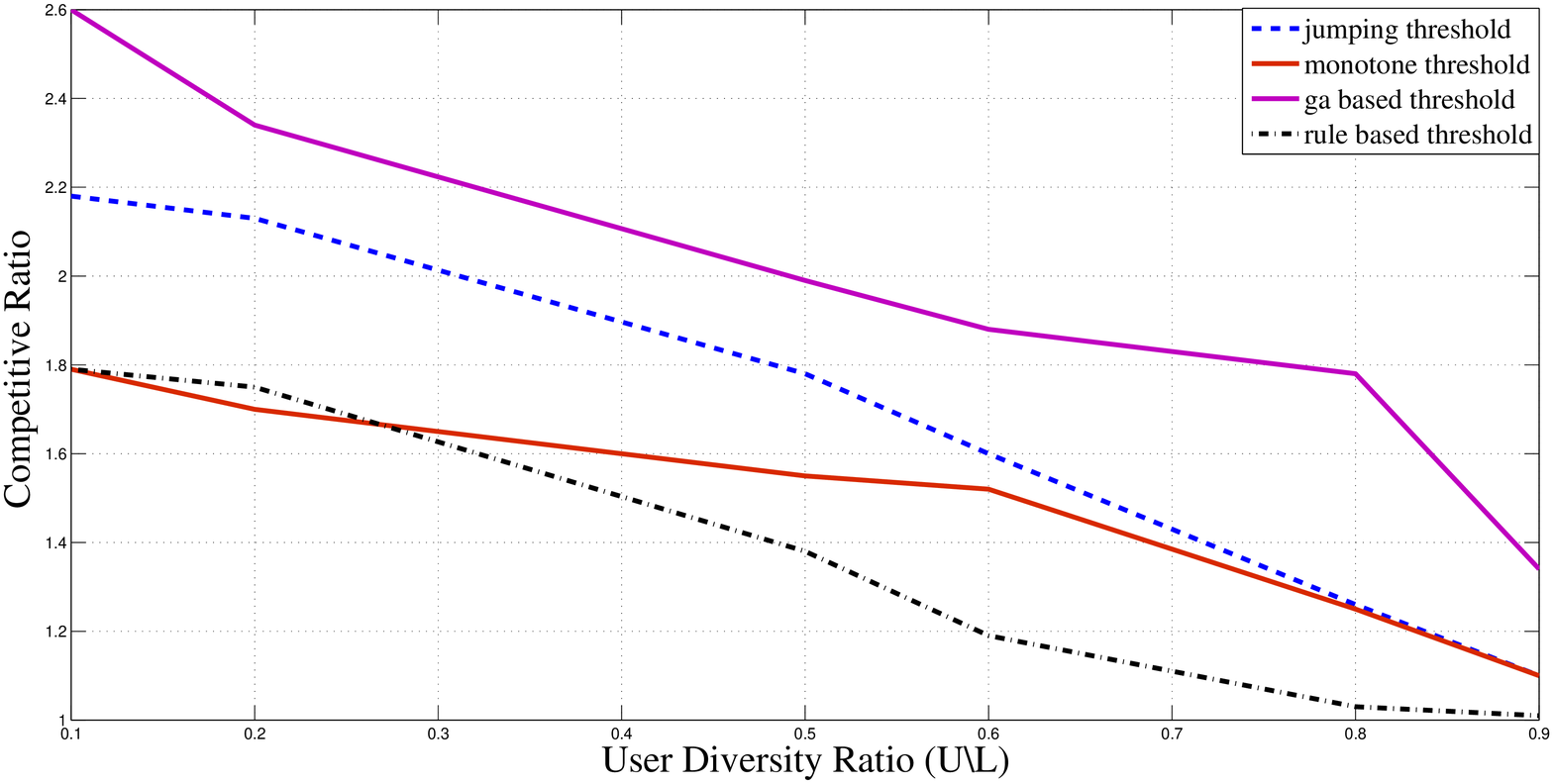}
\centering
\end{psfrags}
\caption{Performance Evaluation of Different Online Threshold Heuristics vs. Diversity in Users Characteristics:  \textit{Worst Case Competitive Ratio Analysis},  Monte Carlo Simulation of $1000$ runs over a Randomly Generated User Sequence of $N=1000$ Users  under $J=10$ Energy Harvests of Different Amounts Modelled over 24-Hour with APOM Capacity Constraint of $2000mJ$}
\label{fig:comp1}
\end{figure*}
\begin{figure*}
\centering
\begin{psfrags}
\includegraphics[scale=0.4]{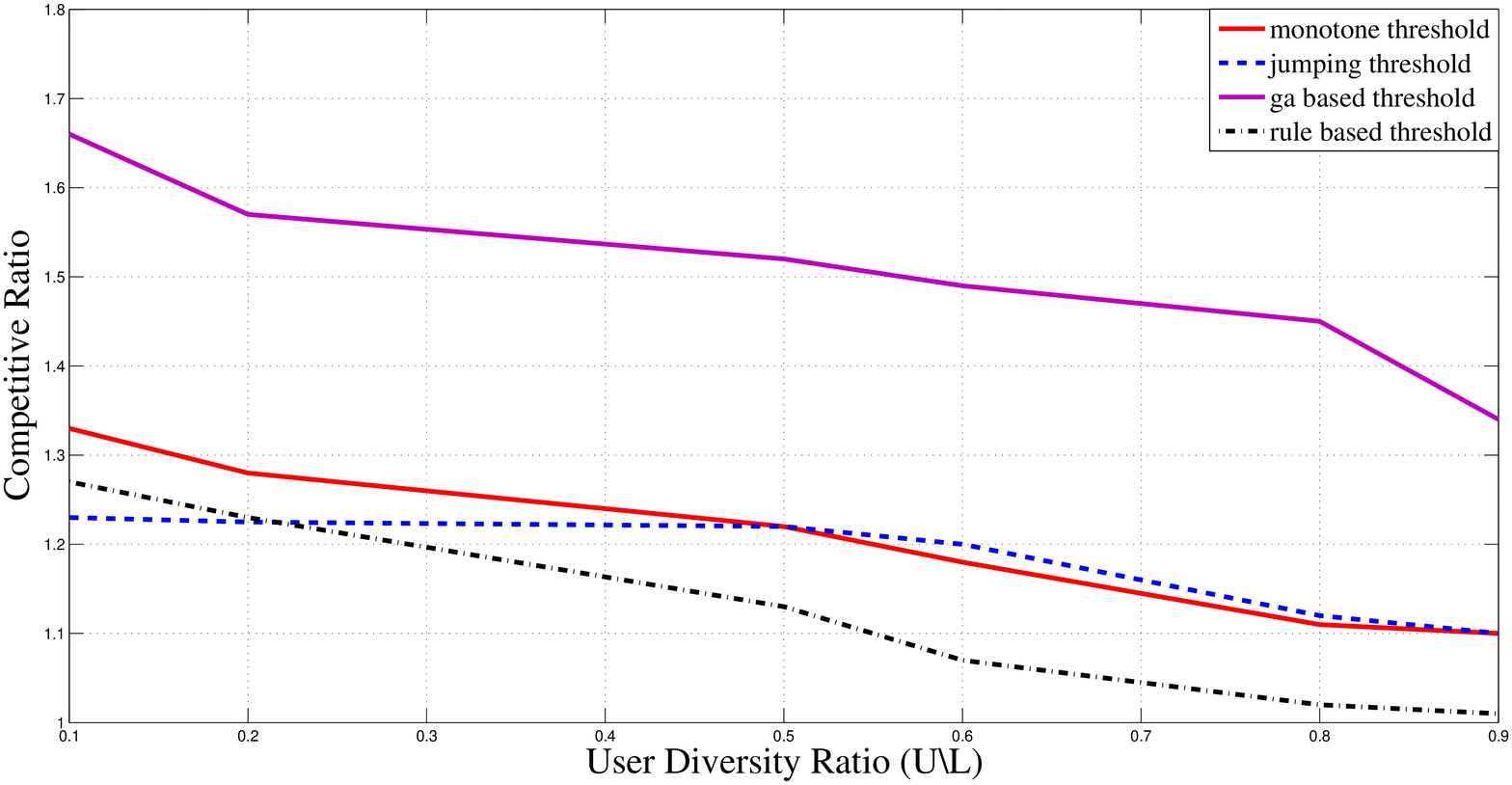}
\centering
\end{psfrags}
\caption{Performance Evaluation of Different Online Threshold Heuristics vs. Diversity in Users Characteristics: \textit{Average Case Competitive Ratio Analysis},   Monte Carlo Simulation of $1000$ runs over a Randomly Generated User Sequence of $N=1000$ Users  under $J=10$  Energy Harvests of Different Amounts Modelled over 24-Hour with APOM Capacity Constraint of $2000mJ$}
\label{fig:comp2}
\end{figure*}

\section{Conclusion}
\label{sec:conc}

In this paper, we addressed the  online user admission problem for an \textit{Access Point On the Move}. Due to the nature of the problem the overall model is structured above a knapsack problem (KP) with dynamic and incremental capacity as the energy of the access point gets replenished at arbitrary time instances while responding to random user demands. We investigated the problem under two different setups where energy and user arrivals are modelled stochastically as well as deterministically. The optimality and structure of a threshold based solution to the stochastic problem was shown, and a computationally friendly ``Expected Threshold Policy" was shown to well approximate the optimal DP solution. On the deterministic side, we considered adaptive threshold based policies  where a user is admitted if its utility to weight ratio exceeds a certain threshold which may be static or dynamic. A competitive ratio was exhibited  for a monotone threshold for the two-harvest scenario. In addition to extended online threshold functions based on  previous literature, threshold functions using Rule Based approach and a Genetic Algorithm are also developed. 

As far as proposed heuristics are concerned, dynamic programming and GA approach propose the best solution if the user  statistics are available beforehand. Since we dealt with developing online user admission mechanism over random user profiles, rule based method yields the best performance using the calibrated rules over  all possible user sequences. Experimental results demonstrate that the proposed decision methods using different threshold functions for the resource allocation problem of APOM are efficient in achieving close to optimal competitive ratios in addition to low computational complexity. 


\section{Appendix}
\label{append}
\begin{proof of Theorem*}
When we restrict attention to the case where
the value/weight values of encountered users are upper and lower bounded by $U, L > 0$, i.e. $L\leq \frac{v}{w}\leq U$, the total amount of energy $B=B_1+B_2+...+B_J$ collected from all harvests are considered as if there is a static energy capacity $B$ then the proof of constant competitive ratio for ``Monotonic Threshold" reduces to the proof of competitive ratio for online knapsack problem with static capacity described in \cite{Zhou2008}. Following the steps given in \cite{Zhou2008}, for any input sequence of $\sigma$ after some time including energy harvests, let the algorithm terminate filling $Z$ fraction of the total capacity (total amount of energy harvests until that instant). Let $S$ and $S^*$ denote the set of selected users by the Monotone Threshold method and the offline optimal algorithm respectively, that is $x_j=1 \mbox{ if } j\in S \mbox{ and } x^*_j=1 \mbox{ if } j\in S^* $. In (\ref{eq:proof1}) and (\ref{eq:proof2})  the weight and value of the  common items in both sets are assigned to the variables $W$ and $V$. Then, the proof of deterministic competitive ratio is given as follows:
\begin{align}
\displaystyle\sum_{j \in (S \cap S^*)}w_j\triangleq W
\label{eq:proof1}\\
\displaystyle\sum_{j \in (S \cap S^*)}v_j\triangleq V
\label{eq:proof2}
\end{align}
An upper bound is needed to be defined on the total value of optimal algorithm. Therefore, since all the users to be selected by the optimal algorithm but not by the Monotone Threshold Algorithm have value over weight ratios smaller than the threshold at that instant and threshold is an increasing function, we have an upper bound as:
\begin{align}
OPT(\sigma)\leq V+\psi(Z)(B-W)\\
\frac{OPT(\sigma)}{A(\sigma)}\leq\frac{V+\psi(Z)(B-W)}{V+v(S\setminus S^*)}
\end{align}
Using the threshold function we may define upper bounds for the common total value parameter $V$ and remaining total value of optimal algorithm as $V_1$ and $V_2$ respectively.
\begin{align}
V\geq\displaystyle\sum_{j \in (S \cap S^*)}\psi(z_j)w_j\triangleq V_1
\\
v(S\setminus S^*)\geq\displaystyle\sum_{j \in (S \setminus S^*)}\psi(z_j)w_j\triangleq V_2
\end{align} Then the competitive ratio can be found as:\begin{align}
\frac{OPT(\sigma)}{A(\sigma)}\leq\frac{V+\psi(Z)(B-W)}{V+v(S\setminus S^*)}
\\\leq\frac{V_1+\psi(Z)(B-W)}{V_1+v(S\setminus S^*)}\leq\frac{V_1+\psi(Z)(B-W)}{V_1+V_2}
\\
\frac{OPT(\sigma)}{A(\sigma)}\leq\frac{\psi(Z)B}{\displaystyle\sum_{i \in S}\psi(z_j)w_j}\leq\frac{\psi(Z)}{\displaystyle\sum_{i \in S}\psi(z_j)\Delta z_j}
\label{eq:pro2}
\end{align}
Then, the assumption of encountering very small weights with respect to the capacity is used and $\Delta z_j$ is defined as follows:
\begin{align}
\Delta z_j \triangleq z_{j+1}-z_j=w_j/B   ~\forall j
\\
\displaystyle\sum_{j \in S}\psi(z_j)\Delta z_j \cong\int_{0}^{Z}\psi(z)dz \\
=\int_{0}^{c}Ldz+\int_{c}^{Z}\psi(z)dz\\
=\frac{L}{e}\frac{(Ue/L)^z}{ln(Ue/L)}=\frac{\psi(z)}{ln(U/L)+1}
\label{pro3}
\end{align}
Finally, when the obtained result of (\ref{pro3}) is substituted for the denominator of (\ref{eq:pro2}), we have a deterministic competitive ratio given as:
\begin{align}
\frac{OPT(\sigma)}{A(\sigma)}\leq ln(U/L)+1
\end{align}
\end{proof of Theorem*}
\enp

\bibliography{apomref_ETT}

\end{document}